\documentclass[11pt]{article}
\usepackage[margin=1in]{geometry}
\usepackage{natbib}
\setcitestyle{open={(},close={)}}
\usepackage[colorlinks=true,citecolor=blue]{hyperref}

\usepackage{mathtools} % amsmath with fixes and additions
\usepackage{booktabs} % commands to create good-looking tables
\usepackage{amsfonts,amsmath,amssymb,amsthm}
\usepackage{algorithm,algorithmic}
\usepackage{tikz} % nice language for creating drawings and diagrams
\theoremstyle{plain}
\newtheorem{theorem}{Theorem}

\newtheorem{lemma}[theorem]{Lemma}
\newtheorem{corollary}[theorem]{Corollary}
\theoremstyle{definition}

\theoremstyle{remark}

\title{\bf  Offline Change Detection under  Contamination}

\author{\vspace{0.5in}\\
\textbf{Sujay Bhatt, Guanhua Fang, Ping Li} \\\\
Cognitive Computing Lab\\
Baidu Research\\
10900 NE 8th St. Bellevue, WA 98004, USA\\\\
\{sujaybhatt.hr, fanggh2018, pingli98\}@gmail.com\\
}

\date{}

\begin{document}
\maketitle

\begin{abstract}\vspace{0.3in}
 \noindent In this work, we propose a non-parametric and robust change detection algorithm to detect multiple change points in time series data under contamination. The contamination model is sufficiently general, in that, the most common model used in the context of change detection -- Huber contamination model -- is a special case. Also, the contamination model is  oblivious and arbitrary. The change detection algorithm is designed for the offline setting, where the objective is to detect changes when all data are received. We only make weak moment assumptions on the inliers (uncorrupted data) to handle a large class of distributions. The robust scan statistic in the algorithm is fashioned using mean estimators based on influence functions. We establish the consistency of the estimated change point indexes as the number of samples increases, and provide empirical evidence to support the consistency results.
\end{abstract}

\newpage

\section{Introduction}
Change point detection in time series data is the task of identifying changes in the underlying data generation model and can be traced back to the initial work of~\cite{Pag54,Pag55} in the context of statistical process/quality control. This simple and elegant framework has been deployed in diverse applications such as bioinformatics~\citep{VB10}, finance~\citep{PS01,PP17}, biology~\citep{Sig13}, climatology~\citep{VHNC10}, metric learning~\citep{LBA14}; to name a few. 

Change detection methods are mainly classified into online and offline settings. In the online setting, the aim is to detect changes as soon as they occur in real-time by optimizing an objective that trades-off detection delay and false alarm; see~\cite{PH08} for a detailed introduction and~\cite{XV21} for a survey of recent developments. In contrast, in the offline setting, the changes need to be detected in a retrospective manner by `segmenting' the entire dataset~\citep{AC17}. Here the objective is to design consistent algorithms and empirically validate using well-known metrics such as F1-Score~\citep{DP20}, Hausdorff metric~\citep{HL10}, etc; see~\cite{TOV20} for detailed overview of the methods and recent developments. 

In this work, we consider the offline setting and contribute to the literature by relaxing the common assumptions. To motivate the setup considered with an example, consider monitoring mean shifts in non-stationary processes using Wireless Sensor Networks (WSN)~\citep{akyildiz2010wireless, cui2019wireless}. In addition to the inherent challenges such as dealing with non-i.i.d data~\citep{Tar19} and heavy-tails~\citep{FR19,bhatt2021extreme}, modern machine learning applications have to deal with the introduction of adversarial examples in the dataset~\citep{KGB18,JL17}. Specifically, when the WSNs are used in  applications such as healthcare (EEG/ fMRI), environmental impact monitoring, energy consumption, etc; the sensors are typically deployed in harsh conditions. This increases data corruption or erroneous readings during
transmission. When the data from a collection of near-by sensors are logged for surveillance and event-classification, any inference procedure should account for the following salient features: \textit{non-i.i.d, outliers, and adversarial contamination}. This motivates the development of change detection algorithms in the offline setting~\citep{AC17,TOV20} that can tackle all the above challenges  in a systematic manner. 

\subsection{Main Results}
We propose a non-parametric change detection algorithm that can deal with non-i.i.d data, outliers, and a weak form of adversarial contamination to identify the change points in a consistent manner. Specifically, we make the following contributions:
\begin{enumerate}
    \item Non-parametric algorithms feature a key quantity known as \textit{scan statistic}, for example CUSUM statistic of~\cite{Pag54}, which is required to `scan' the dataset to identify the change points. We propose a scan statistic based on influence functions proposed by~\cite{Cat12} that can handle outliers and heavy-tails, to deal additionally with contamination. We consider a contamination model, where the
    outliers (corrupted data) are correlated to each other and to inliers (uncorrupted data). The inliers can also be correlated to one another. The resulting robust  non-parametric algorithm \texttt{RC-Cat} announces a change if the scan statistic exceeds a pre-specified threshold, provided the scan statistic is a local maximum. This additional sophistication of local search methods was the introduced in~\cite{NZ12} and developed for the robust version in~\cite{LY21}, to mainly avoid overestimation~of~change~points. 
    
    \item A natural way to theoretically evaluate change detection algorithms is to establish consistency of the estimated change point indexes as the number of samples increases. In particular, we show that \texttt{RC-Cat} is consistent in the presence of contamination, i.e, as the number of data points~$n \uparrow \infty$,
    \[
    \mathbb{P}\Big(\widehat K = K, \max_{k=1}^{\widehat K} |\widehat \tau_k - \tau_k| \leq w \Big) \rightarrow 1,
    \]
    where~$K$ is the number of true change points located at~$\tau_k,~k \in \{1,\cdots, K\}$ and $\widehat K$ is the number of detected change points announced at~$\widehat \tau_k,~k \in \{1,\cdots,\widehat K\}$, and~$w>0$ is parameter that related to the window length considered.
\end{enumerate}

\subsection{Related Literature}
In the context of robust change detection, a common model of contamination that is considered to design algorithm is the Huber contamination model~\citep{Hub64}. In this model, the data generation model is a mixture model~$(1-\eta)F+ \eta Q$, where~$F$ is the true distribution before the change and~$Q$ is any arbitrary distribution with a probability~$\eta$. Using such a model of contamination,~\cite{Hus13} make use of M-estimation idea from robust statistics~\citep{Hub04} to address the change point detection problem in the context of regression.~\cite{FR19} consider penalized M-estimation based procedure that can deal with outliers.~\citep{PBR20} showed that Huber contamination model is equivalent to assuming a heavy-tailed noise for the i.i.d data. In light of this,~\cite{YC22} propose a scan statistic based on U-statistics to deal with heavy-tailed noise distributions. The setup and analysis considered in this work is closest to~\cite{LY21}, however, with the following key differences:
\begin{itemize}
    \item~\cite{LY21} consider change detection under i.i.d data. While this is a useful starting point, it only serves as a crude approximation when the data is gathered from heterogeneous sources~\citep{MMS08} and is in general non-i.i.d~\citep{Tar19}.
    \item The scan statistics is fashioned using the robust estimator (RUME) in~\cite{PBR20}. RUME uses half of the samples to identify the shortest interval containing at least $(1- \eta)n$ fraction of the points, and then the remaining half of the points is used to return an estimate of the mean. While this is acceptable in the case of robust mean estimation, it has clear disadvantages in the context of change point detection, where the initial segregation might hide/ remove the true change points. Another feature of RUME is that the amount of contamination that the estimator can handle is limited, and this limits the applicability in many applications. 
\end{itemize}
In contrast, our algorithm deals with non-i.i.d inliers and contamination, where the inliers only have a bounded second moment. Also, unlike~\cite{LY21}, we do not segregate the data for robust mean estimation, which avoids the problem of loosing change points. Using empirical results, we further show that, not only the proposed algorithm is more general than that in~\cite{LY21}, it is faster and obtains better detection performance across different settings. 

\newpage

\section{Mean Estimation under Contamination}
In this section, we propose a robust mean estimator that can deal with  \textbf{non-i.i.d data} with arbitrary contamination. The estimator is based on influence functions proposed in~\cite{Cat12} and~\cite{CG17}, and is further developed in~\cite{BFLG2022a, BFLG2022b}. Let~$\{X_t\}_{t=1}^n$ be a collection of real-valued random variables. Let~$\mathcal{F}_0$ denote the trivial sigma algebra, and let~$\mathcal{F}_t$ denote the sigma-algebra generated by the set~$\{X_1,X_2,\cdots,X_t\}$, whence~$X_{t}$ is~$\mathcal{F}_t-$measurable. Let~$[n]:=\{1,2,\cdots,n\}$.
\begin{itemize}
    \item[\textbf{C1.}] The set $\{X_t\}_{t \in [n]}$ is such that the (unknown) conditional expected value
    \[
     \forall~t \in [n],~~ \mathbb{E}\Big[X_t | \mathcal{F}_{t-1} \Big] = \mu_t.
    \]
    \item[\textbf{C2.}] The conditional second moment of~$X_t$ is bounded, i.e, for a known~$\mathcal{M}>0$,
    \[
      \forall~t \in [n],~~    \mathbb{E}\Big[X_t^2 | \mathcal{F}_{t-1} \Big] \leq \mathcal{M}.
    \]
\end{itemize}
It is easy to see that i.i.d is a special case that satisfies \textbf{C1} and \textbf{C2}. However, the model allows for more general dependencies, see~\cite{SLC12}.

\subsection{Contamination Model}
We assume that for some corruption rate $0<\eta<1$, an adversary may change at most~$\eta k$ of any sub-sequence of $\{X_t\}_{t \in [n]}$ with length $k$ at least $k \geq k_0$, to arbitrary values. The resulting set of observations will be $\widetilde X_1,\widetilde X_2, \cdots, \widetilde X_n$, so that
\begin{equation} \label{e:corrupt}
  \sup_{i \in [n-k]} \sum_{j=1}^k \mathcal{I}\Big( \widetilde X_{i+j}\not= X_{i+j} \Big) \leq \eta k,
\end{equation}
where~$\mathcal{I}(\cdot)$ denotes the indicator function and $k \geq k_0$ with $k_0$ being a fixed integer such that $k_0 = \Omega(\log n)$. 
The task is to estimate the true mean $\mu:= \frac{1}{n} \sum_{i=1}^n \mu_i$ based on the observations~$\widetilde X_1,\widetilde X_2,\cdots, \widetilde X_n$. This contamination model is widely-studied in machine learning for i.i.d data; see~\citet{CSV17,HL18,LM21} and the related references for existing results. The contamination model is similar to~$\mathcal{I} \cup \mathcal{O}$ model of~\cite{LL19} and also shares similarities with the Huber contamination model~\citep{Hub64}. While the contamination can be arbitrary, we do not allow the possibility where the adversary corrupts a fraction of the sample possibly with the knowledge of the whole dataset to intentionally hide the change points, i.e, the contamination is \textit{weakly adversarial}. 
%This is a reasonable assumption in the context of change detection, as removal of change points simply invalidates the main objective of the detector. Without additional side information, we believe, a strong contamination model as in~\cite{DKKLMS17} is not a useful object of study for the change detection problem. 

\textbf{Remark}.~ The well-known Huber contamination model in change detection~\citep{LY21} is a special case of the considered adversarial model. Let~$\epsilon \in [0,1]$ denote the outlier distribution probability in the Huber contamination model, i.e, the data is generated as~$(1-\epsilon) P + \epsilon Q$, where $P$ is the true distribution and~$Q$ is any arbitrary distribution. Let the empirical fraction~$\widehat{\epsilon}_n = \sup_{i \in [n-k]} \frac{1}{k}\sum_{j=1}^k \mathcal{I} \bigl( \tilde X_{i+j}\not= X_{i+j}\bigr)$. According to a recent result in~\cite{BFLG2022b}, with probability at least~$1 - \beta$, we have for all~$k$
\begin{equation*}
    \widehat{\epsilon}_n \leq \epsilon + \underbrace{1.7 \sqrt{\epsilon (1-\epsilon)} \sqrt{\frac{\log(\log(2n)) + 0.72 \log\frac{ 10.4 n}{\beta}}{k}}}_{:=f(\beta,\epsilon,k)}.
\end{equation*}
Fix~$k_0> c \log (n/\beta)$ and set 
\begin{equation*} 
\epsilon + f(\beta,\epsilon,k_0) =: \eta.
\end{equation*}
Clearly, for all~$k \geq k_0$, we have the corruption fraction~$\widehat{\epsilon}_n$ to be at most~$\eta$ with a very high probability. 

\subsection{Mean Estimation with Influence Functions} \label{sec:Mean.IF}
The idea of using influence functions for robust mean estimation is not new~\citep{Hub04}. However these M-estimators are unable to scale gracefully with dimension~\citep{Mar76,DG92}, and~\cite{PBR20} show that the bias scales polynomially with dimension. This led to the development of a class of M-estimators introduced by~\cite{Cat12} that can be used to obtain dimension-free bounds in the vector settings~\citep{CG17}. With a similar future objective in mind, we make use of the influence functions proposed in these works to fashion a robust estimator that has minimax optimal asymptotic bias in the contamination parameter when the data sequence is more general than~$i.i.d$. 

Consider a non-decreasing function $\psi:\, \mathbb{R} \rightarrow \mathbb{R}$ such that
\begin{equation*} 
  -\log(1-x+x^2/2)\leq \psi(x)\leq \log(1+x+x^2/2)
\end{equation*}
for all $x\in \mathbb{R}$ as in~\cite{Cat12}. One can choose such a function that is bounded:
specifically, we assume that for some $0<A<\infty$, 
\begin{equation} \label{e:bound}
  |\psi(x)|\leq A \ \ \text{for all} \ \ x\in \mathbb{R}.
\end{equation} 
From~\cite{Cat12}, the narrowest possible choice for the influence function has $A=\log
2$, and is given by
\begin{equation} \label{eq:psi.cat}
\psi(x)=
\begin{cases}
-\log(1-x+x^2/2),\hfill~~0\leq x\leq 1,\\
~\log(2),\hfill~~x\geq 1,\\
-\psi(-x),\hfill~~-1 \leq x \leq 0, \\
-\log(2),\hfill~~x\leq -1.
\end{cases}
\end{equation}
We consider an estimator based on soft-truncation after re-scaling, defined by
\begin{equation} \label{eq:rob.est}
    \widehat \mu_{\eta}:= \frac{\alpha}{n} \sum_{i=1}^n \psi(\frac{\widetilde X_i}{\alpha}),
\end{equation}
where~$\alpha>0$ is a re-scaling parameter, and the uncontaminated version is given as
\begin{equation} \label{eq:emp.mean}
    \widehat \mu:= \frac{\alpha}{n} \sum_{i=1}^n \psi(\frac{X_i}{\alpha}).
\end{equation}
In the absence of contamination, depending on the choice of~$\psi(\cdot)$ and~$\alpha$, the estimator~\eqref{eq:emp.mean} can closely approximate the empirical mean; see~\cite{Hol19} for example. Similar estimator for i.i.d data was considered in~\cite{Hol19}, where the deviation bounds were characterized using well-known PAC Bayesian inequalities inspired by Donsker-Varadhan's variational formula~\citep{Cat04,DE11}. However, since the data are not i.i.d in our case, we need a different approach to characterize the deviations, and this is the main contribution of this section.

\begin{theorem} \label{thm:dev.bnd}
Consider a collection of random variables~$\{ \widetilde X_t\}_{t \in [n]}$. Let~$\alpha = \sqrt{\frac{\mathcal{M}}{2\Big( \frac{\log(2/\delta)}{n} + 2 A \eta \Big)}}$ and $\delta \in (0,1)$. The estimator~\eqref{eq:rob.est} satisfies 
\begin{equation} \label{eq:dev.bnd}
    |\widehat \mu_{\eta} - \mu| \leq \sqrt{2 \mathcal{M} \Big( \frac{\log(2/\delta)}{n} + 2 A \eta \Big)},
\end{equation}
with probability at least~$1 - \delta$.
\end{theorem}
A high-probability deviation bound for~$\widehat \mu$, i.e, in the absence of contamination is first characterized as a function of~$\alpha$, whence we obtain
\[
|\widehat \mu - \mu| \leq \frac{\mathcal{M}}{2\alpha} + \frac{\alpha \log(2/\delta)}{n}.
\]
From~\eqref{e:bound}, we have the following relation
\begin{equation*} 
    |\widehat \mu_{\eta} - \mu| \leq |\widehat \mu - \mu| + 2A\eta \alpha.
\end{equation*} 
This provides the deviation bound of the overall soft-truncation estimator~\eqref{eq:rob.est}. 

\begin{corollary} \label{cor:bnd.eta}
Let~$B>1$. Under the assumptions as in Theorem~\ref{thm:dev.bnd} such that~\eqref{eq:dev.bnd} holds, we have with probability at least~$1 - 2\exp{\Big(- \frac{A\eta}{B} n\Big)}$
\begin{equation}
    |\widehat \mu_{\eta} - \mu| \leq c_0 \sqrt{\mathcal{M} \eta}, 
\end{equation}
where~$c_0^2:= 2A (1/B + 2)$. 
\end{corollary}

Corollary~\ref{cor:bnd.eta} obtains the deviation purely in terms of the contamination fraction, and
will be useful later in establishing the consistency of change detection algorithms. Another useful feature is that it informs the choice of segmentation window length that guarantees a tight deviation characterization. 

From~\eqref{eq:dev.bnd}, it is clear that there is an asymptotic ($n \uparrow \infty$) bias  of~$O(\sqrt{\mathcal{M} \eta})$ associated with the estimator owing to contamination. Also, when~$\mu_t \in [0,1]$ with a possibly heavy tail martingale difference noise-- a common assumption in bandits~\citep{LS20} and reinforcement learning~\citep{AJK19}-- the deviation bound and hence the bias can be written in terms of the (conditional) variance~$\sigma^2$ as~$O(\sigma \sqrt{\eta})$ by using the standard~$C_{r}$ inequality{\footnote{For any random variable~$Y$, real number~$\gamma$, and~$r>0$,
\[
|Y|^r \leq \max\{2^{r-1},1\} (|Y-\gamma|^r+|\gamma|^r).
\]}}.
This matches the minimax lower bound~\citep{DKKLMS17,HL18} that is shown to be information theoretically optimal. 

However, in general, as the deviation~\eqref{eq:dev.bnd} depends on the non-centered moment, it is sensitive to the location of the distribution.~\cite{CG17} propose a `shifting-device' approach to obtain centered estimates that can be used to obtain a deviation bound in terms of the conditional variance. This has been used for PAC-Bayesian analysis using influence functions in~\cite{Hol19}.

\newpage

\begin{theorem} \label{thm:cnt.est}
Consider the set of r.vs~$\{\widetilde X_t\}_{t \in [n]}$ such that~$\mu_t= \mu,~\forall~t$. Additionally, let~$\mathcal{V}$ denote an upper bound on conditional variance of the uncontaminated random variables. Let~$0<k<n$ denote the length of the data to create a shifting device. Let~$\alpha = \sqrt{\frac{(\mathcal{V} + \vartheta_k^2)}{2\Big( \frac{\log(2/\delta)}{n-k} + 2 A \eta \Big)}}$ with~$\vartheta_k = \sqrt{2 \mathcal{M} \Big( \frac{\log(2/\delta)}{k} + 2 A \eta \Big)}$. The estimator~\eqref{eq:rob.est} satisfies
\begin{equation}~\label{eq:dev.bnd.cent}
    |\widehat \mu_{\eta} - \mu| \leq \sqrt{2 (\mathcal{V}+ \vartheta_k^2 )\Big( \frac{\log(4/\delta)}{n-k} + 2 A \eta \Big)},
\end{equation}
with probability at least~$1 - \delta$.
\end{theorem}

Theorem~\ref{thm:cnt.est} provides a deviation bound as a function of the conditional variance. When the contamination level~$\eta$ is negligible, a judicious choice of~$k$ will lessen the dependence on the raw moments, and the conditional variance in the deviation term. Theorem~\ref{thm:cnt.est} works to combat sensitivity to the distribution location. A procedure to obtain an estimator having the deviation bound as in~\eqref{eq:dev.bnd.cent} is given as follows:
\begin{itemize}
    \item[i.)] Shifting-Device: Let~$\{\widetilde X_i\}_{i=1}^k$ denote a sub-set of the collection. Compute a soft-truncated estimate using these~$k$ samples,
    \[
    \Bar{\mu}_{\eta}:= \frac{\Bar{\alpha}}{k} \sum_{i=1}^k \psi(\frac{\widetilde X_i}{\Bar{\alpha}}),
    \]
    where~$\Bar{\alpha}$ is informed by Theorem~\ref{thm:dev.bnd}.
    \item[ii.)] Shift the remaining~$n-k$ samples by~$\Bar{\mu}_{\eta}$, i.e,~$\widetilde X'_i = \widetilde X_i - \Bar{\mu}_{\eta}$, whence the conditional second moment of this data is now bounded by~$(\mathcal{V}+ \vartheta_k^2 )$. Computing the soft-truncated estimate of this data 
    \[
    \mu'_{\eta} := \frac{\alpha'}{n-k} \sum_{i=k+1}^n \psi(\frac{\widetilde X'_i}{\alpha'}),
    \]
    where~$\alpha'$ is informed by Theorem~\ref{thm:cnt.est}.
    \item[iii.)] Estimator~$\widehat \mu_{\eta} = \mu'_{\eta} +  \Bar{\mu}_{\eta}$ has the desired properties.
\end{itemize}

\section{Offline Change Detection}
In the rest of the paper, we assume that the contamination model used by the adversary is as in~\eqref{e:corrupt}. We first provide an algorithm based on the robust estimation techniques discussed in Section~\ref{sec:Mean.IF}, and later establish the theoretical properties of the algorithm. 

\subsection{The Proposed Algorithm}

Algorithm~\ref{alg:seq} is an offline robust change detection algorithm that can handle~$\eta$ fraction of weakly adversarial contamination when the data is not necessarily i.i.d. The methodology is inspired by~\cite{NZ12} and~\cite{LY21}, which handle the uncontaminated and weak contamination situations respectively.

\newpage

\begin{algorithm}[h]
\begin{algorithmic}[1]
\STATE \textbf{Input:} $\{\widetilde X\}_{i=1}^n$, $b (\text{threshold}) > 0,2w
(\text{window})>0$, $\eta \in (0,1)$, $\lambda \geq 1$
\STATE $\mathcal{L} \leftarrow \emptyset$, $\mathcal{G} \leftarrow \emptyset$
\FOR{$j \in \{w+1, \cdots, n - w\}$}
\STATE $S_{w}(j) \leftarrow \Big| \Psi\Big(\{\widetilde X_i\}_{i = j+1}^{j+w}\Big) - \Psi\Big(\{\widetilde X_i\}_{i = j - w}^{j-1}\Big) \Big|$ 
\ENDFOR
\FOR{$j \in \{\lambda w+1, \cdots, n - \lambda w\}$}
\IF{$j$ is a $\lambda w-$local maximizer of $S_{w}(j)$}
\STATE $\mathcal{L} \leftarrow \mathcal{L} \cup \{j\}$
\ENDIF
\ENDFOR
\FOR{$k \in \mathcal{L}$}
\IF{$S_w(k) > b$}
\STATE $\mathcal{G} \leftarrow \mathcal{G} \cup \{k\}$
\ENDIF
\ENDFOR
\STATE \textbf{Output:} $\mathcal{G}$
\end{algorithmic}
\caption{Robust Change Detection with Catoni (\texttt{RC-Cat})}
\label{alg:seq}
\end{algorithm}

Algorithm~\ref{alg:seq} is an intuitive solution that combines local and global search methods in a non-parametric manner to identify the change points. It works as follows: The dataset is scanned using the scan statistic~$S_{w}(\cdot)$, which is the absolute difference between the robust estimates of mean over specified length~$w$. Here the estimator over length~$w>0$,
\[
\Psi(\{\widetilde X_i\}_{i=1}^w):= \frac{\alpha}{w} \sum_{i=1}^w \psi\Big(\frac{\widetilde X_i}{\alpha}\Big),
\]
with one possible choice of~$\psi(\cdot)$ given by~\eqref{eq:psi.cat}. 
The nature of the (non-parametric) scan statistic, where normalized estimates of equal length of samples are compared, is well-studied in the literature. For example,~\cite{CWKX19} make use of similar ideas for empirical means of independent sub-gaussian distributions to detect changes in the mean in multi-armed bandit problems, while~\cite{NZ12} consider an application in bioinformatics. The robust scan statistic is closest to that in~\cite{LY21}, except with a few key differences: (i)~There is no sample splitting to estimate the location parameter using RUME~\citep{PBR20}. In~\cite{LY21}, the data over~$w/2$ is used to simply identify a high-confidence interval, and the remaining~$w/2$ portion is used to calculate the robust mean. This not only increases the variance of the estimator, but also may hide/ remove change points depending on which of~$w/2$ points is selected. This affects the detection delay and hence the consistency. (ii)~The worst case computational complexity of \texttt{RC-Cat} is~$O(nw)$, whereas the worst case complexity in case of~\cite{LY21} is~$O(n^2 w \log(w))$. Here the~$w \log(w)$ is from ranking the data to find the shortest interval involved in RUME. Note that the state-of-the-art methods such as penalized bi-weight loss methods have a computational complexity of~$O(n^3)$~\citep[Corollary 2]{FR19}.  

The~$\lambda w-$local maximizer is inspired by~\cite{NZ12} and also appears in~\cite{LY21}. $S_w(j)$ is a $h-$local maximizer if~$S_w(j) \geq S_w(k)$ for all~$k \in (j-h,j+h)$, and is motivated by two key ideas: (i)~It helps to avoid overestimating the number of change points. (ii)~It helps to localize change points with a high probability. 

\subsection{The Analysis}
\texttt{RC-Cat} is a computationally appealing solution to offline change detection. In this section, we establish that it is consistent as well, i.e, as the number of data samples increase, the regime changes are identified within a prescribed margin with a high probability. We need to make a few standard assumptions to enable this result. 

Let~$K$ be the number of true change points located at~$\tau_k,~k \in \{1,\cdots, K\}$ with~$\tau_0=0$ and~$\tau_{K+1} = n$. Let $\widehat K$ be the number of detected change points announced at~$\widehat \tau_k,~k \in \{1,\cdots,\widehat K\}$ by \texttt{RC-Cat}. Let the collections be denoted as~$\mathcal{K}:= \{\tau_1,\cdots,\tau_{K}\}$ and $\widehat{\mathcal{K}}:= \{\widehat \tau_1, \widehat \tau_2,\cdots, \widehat \tau_{\widehat K}\}$ respectively. Let the minimal spacing be denoted as $\delta = \min_{k \in \mathcal{K}}|\tau_{k+1} - \tau_{k}|$ and the jump size be denoted as $\theta = \min_{k \in \mathcal{K}} |\mu_{\tau_{k}+1} - \mu_{\tau_k}|$.

\begin{itemize}
    \item[\textbf{A1.}] The conditional expectation in~\textbf{C1} is constant between change points, that is,  \\ for $k \in \{1,2,\cdots,K+1\}$,
    \[
    \mu_t = \mu_{\tau_{k-1}},~~\forall~{\tau_{k-1}} \leq t < \tau_k.
    \]
    \item[\textbf{A2.}] The minimal spacing~$\delta > \lambda w$ for some~$\lambda \geq 2$.
    \item[\textbf{A3.}] The jump size~$\theta > \sqrt{3 b}$, where~$b$ is the threshold.
\end{itemize}
\textbf{A2} essentially says that the process has slow changes and \textbf{A3} is related to detectability. The assumptions \textbf{A2} and \textbf{A3} are intuitive and standard in the change detection literature~\citep{NZ12,CWKX19,Yu20,LY21}, while \textbf{A1} simplifies exposition. While these conditions are necessary for characterizing the theoretical properties, deviations from these assumptions do not drastically affect the empirical performance. Also, we should mention that, \textbf{A1} can be relaxed to allow small perturbations between change points for the conditional expectation, and the same analysis carries over. 

\begin{theorem} \label{thm:const.alg}
Let~$\{\widetilde X_i\}_{i \in [n]}$ be the collection of r.vs input to \texttt{RC-Cat}. Let the threshold~$b = 2c_0 \sqrt{\mathcal{M} \eta}$
and window ~$w \geq c_1 \log(n) / \eta$. Under assumptions~\textbf{A1} - \textbf{A3}, it holds that
\begin{equation} \label{eq:chg.pt}
    \mathbb{P}\Big(\widehat K = K, \max_{k=1}^{\widehat K} |\widehat \tau_k - \tau_k| \leq w \Big) \geq 1 - n^{-c_1}.
\end{equation}
\end{theorem}
Theorem~\ref{thm:const.alg} shows that for large dataset, \texttt{RC-Cat} identifies the change points or the segments to within specified tolerance with a high probability. In other words, as~$n \uparrow \infty$, we have that 
 \[
    \mathbb{P}\Big(\widehat K = K, \max_{k=1}^{\widehat K} |\widehat \tau_k - \tau_k| \leq w \Big) \rightarrow 1.
 \]
Due to the nature of the robust estimator used in the scan statistic, \texttt{RC-Cat} can handle data from heterogeneous sources and non-i.i.d as specified by~$\textbf{C1}$ and $\textbf{C2}$.

\begin{corollary} \label{cor:no.fa}
Let~$\mathcal{K} = \emptyset$. Under the same assumptions as in Theorem~\ref{thm:const.alg}, there is a constant~$c_1>0$ such that
\[
\mathbb{P}(\widehat K = 0) \geq 1 - n^{-c_1}.
\]
\end{corollary}
Corollary~\ref{cor:no.fa} says that when~$\mathcal{K} = \emptyset$, the proposed algorithm \texttt{RC-Cat} is still consistent, and provides an upper bound on the false detection probability.

\section{Proofs of Main Results}

\vspace{-2mm}

In this section, we provide the proofs of the main results. The proof of Theorem~\ref{thm:dev.bnd} builds on the standard martingale analysis~\citep{Fre75,SLC12} to establish the bounds for bounded functions of real-valued random variables. The key idea is to make use of the fact that the influence function~$\psi(\cdot)$ is bounded by logarithmic functions, and to construct a supermartingale as a function of~$\psi(\cdot)$. The result then follows using Markov's inequality. The proof of Theorem~\ref{thm:const.alg} closely follows~\cite{NZ12}, but under weaker assumptions on the data and the parameters. Also, in comparison with~\cite{LY21}, for a fixed confidence~$\delta$, \texttt{RC-Cat} achieves consistency even over smaller sized datasets.

\subsection{Proof of Theorem~\ref{thm:dev.bnd}}
Before establishing the result, we will first characterize the high-probability deviation bound for the robust estimator in the absence of contamination as a function of~$\alpha$. This is given as Lemma~\ref{lem:dev.emp}.

\begin{lemma} \label{lem:dev.emp}
Let the set of r.vs~$\{X_t\}_{t \in [n]}$ satisfy~\textbf{C1} and \textbf{C2}. For~$\alpha >0$ and $\delta \in (0,1)$, the estimator~\eqref{eq:emp.mean} satisfies with probability at least~$1 - \delta$
\begin{equation*} 
    |\widehat \mu - \mu| \leq \frac{\mathcal{M}}{2\alpha} + \frac{\alpha \log(2/\delta)}{n}.
\end{equation*}
\end{lemma}

\begin{proof}
For any~$t \leq n$, we have the following using the upper bound on the influence function~$\psi(\cdot)$,
\begin{align*}
    \mathbb{E}\Big[\exp{\Big(\psi\Big(\frac{X_t}{\alpha}\Big)\Big)} \Big| \mathcal{F}_{t-1}\Big] &\leq \mathbb{E}\Big[ 1 + \frac{X_t}{\alpha} + \frac{X_t^2}{2\alpha^2} \Big| \mathcal{F}_{t-1} \Big], \\
    & \leq 1 + \frac{\mu_t}{\alpha} + \frac{1}{2 \alpha^2} \mathbb{E}\Big[X^2_t \Big| \mathcal{F}_{t-1} \Big].
\end{align*}
Using the fact that~$1+x \leq e^x$ for all~$x \in \mathbb{R}$, we have using~$\textbf{C2}$
\begin{align} \label{eq:up.bd.exp}
    \mathbb{E}\Big[\exp{\Big(\psi\Big(\frac{X_t}{\alpha}\Big)\Big)} \Big| \mathcal{F}_{t-1}\Big] &\leq \exp{\Big(\frac{\mu_t}{\alpha} + \frac{\mathcal{M}}{2 \alpha^2}  \Big)}.
\end{align}
Construct a sequence of random variables~$Y_t$ as follows:~$Y_0 = 1$ and for~$t \geq 1$, 
\begin{align*}
    Y_t = Y_{t-1} \exp{\Big(\psi\Big(\frac{X_t}{\alpha}\Big)\Big)} \exp{\Big(-\Big(\frac{\mu_t}{\alpha} + \frac{\mathcal{M}}{2 \alpha^2} \Big) \Big)}.
\end{align*}
Clearly,~$\mathbb{E}\Big[ Y_t \Big| \mathcal{F}_{t-1}\Big] \leq Y_{t-1}$ as~$Y_{t-1}$ is~$\mathcal{F}_{t-1}$ measurable and~\eqref{eq:up.bd.exp} holds. We have that the unconditional expectation
\[
\mathbb{E}[Y_n] \leq \mathbb{E}[Y_1] \leq \cdots \leq \mathbb{E}[Y_0]  = 1.
\]
Recursively,~$Y_n$ is expressed as
\begin{align*}
    Y_n &= \exp{\Big( \sum_{t=1}^n \psi\Big(\frac{X_t}{\alpha}\Big)\Big)} \exp{\Big(-\Big(\frac{n \mu}{\alpha} + \frac{n \mathcal{M}}{2 \alpha^2} \Big) \Big)},\\
    &= \exp{\Big( \frac{n \widehat \mu}{\alpha} \Big)} \exp{\Big(-\Big(\frac{n \mu}{\alpha} + \frac{n \mathcal{M}}{2 \alpha^2} \Big) \Big)}.
\end{align*}
Here~$\widehat \mu$ is given as in~\eqref{eq:emp.mean} and $\mu = \frac{1}{n} \sum_{t=1}^n \mu_t$. By Markov's inequality, we have that
\begin{align*}
    \mathbb{P}(Y_n \geq 2/\delta) \leq \frac{\delta \mathbb{E}[Y_n]}{2} \leq \frac{\delta}{2}.
\end{align*}
In other words, we have that
\begin{align*}
     \mathbb{P}\Big( \frac{n \widehat \mu}{\alpha} \geq \frac{n \mu}{\alpha} + \frac{n \mathcal{M}}{2 \alpha^2} + \log(2/\delta) \Big) \leq \frac{\delta}{2}.
\end{align*}
Dividing by~$\frac{n}{\alpha}$ gives the deviation in one direction. Using the lower bound on the influence function, we have that
\begin{align*} 
    \mathbb{E}\Big[\exp{\Big(-\psi\Big(\frac{X_t}{\alpha}\Big)\Big)} \Big| \mathcal{F}_{t-1}\Big] &\leq \exp{\Big(-\frac{\mu_t}{\alpha} + \frac{\mathcal{M}}{2 \alpha^2}  \Big)}.
\end{align*}
Analogous arguments establish the deviation of the estimator~$\widehat \mu$ in the other direction, whence
\begin{align*}
     \mathbb{P}\Big( \frac{n \mu}{\alpha} - \frac{n \widehat \mu}{\alpha} \geq   \frac{n \mathcal{M}}{2 \alpha^2} + \log(2/\delta) \Big) \leq \frac{\delta}{2}.
\end{align*}
The result follows.
\end{proof}

From~\eqref{e:bound}, we have the following relation
\begin{equation} \label{eq:cont.dev}
    |\widehat \mu_{\eta} - \mu| \leq |\widehat \mu - \mu| + 2A\eta \alpha.
\end{equation} 
From~\eqref{eq:dev.bnd} and~\eqref{eq:cont.dev}, we have with probability at least~$1-\delta$,
\[
|\widehat \mu_{\eta} - \mu| \leq \frac{\mathcal{M}}{2\alpha} + \frac{\alpha \log(2/\delta)}{n} + 2A\eta \alpha.
\]
The main result holds by setting~$\alpha$. 

\subsection{Proof of Corollary~\ref{cor:bnd.eta}}
From Theorem~\ref{thm:dev.bnd}, we have with probability at least~$1-\delta$
\[
|\widehat \mu_{\eta} - \mu| \leq \sqrt{2 \mathcal{M} \Big( \frac{\log(2/\delta)}{n} + 2 A \eta \Big)}.
\]
By choosing~$n \geq \frac{B}{A\eta} \log(2/\delta)$, we have the result. 

\subsection{Proof of Theorem~\ref{thm:const.alg}}
The proof is established using the following reasoning. Let~$T:= \{x: |x-\tau_k| > w,~\forall~k \in \mathcal{K} \}$ denote the set of all points that are at least~$w-$away from the true change points. 
Consider the following events,
\begin{align*}
    \mathcal{E}_1(x) &= \{S_{w}(x) < b \}, \\
    \mathcal{E}_2(y) &= \{S_{w}(y) > b \}, \\
    \mathcal{E}_n &= \Big(\cap_{k=1}^K \mathcal{E}_2(\tau_k) \Big) \cap \Big(\cap_{x \in T} \mathcal{E}_1(x) \Big).
\end{align*}
Here~$\mathcal{E}_1(x)$ captures the events that false detection was not raised,~$ \mathcal{E}_2(y)$ captures all the events when the algorithm raised a detection, and~$\mathcal{E}_n$ captures the event that detection was raised only around the region where the true changes occurred. The result holds if we establish two relations
\begin{align*}
    \text{On event}~\mathcal{E}_n,&~\text{we have} \\
    &~\widehat K = K~~\&~\max_{k \in \widehat{\mathcal{K}}}|\widehat \tau_k - \tau_k| \leq w,~\text{and}\\
    &\mathbb{P}(\mathcal{E}^c_n) \rightarrow 0.
\end{align*}
We begin by characterizing the probability of each event as separate results to highlight the assumptions required, and the main result follows from Lemmas~\ref{lem:fir}-\ref{lem:fin}.

\begin{lemma} \label{lem:fir}
Let~$\{\widetilde X_i\}_{i \in [n]}$ be a collection that is input to \texttt{RC-Cat}. Let the threshold~$b = 2c_0 \sqrt{\mathcal{M} \eta}$. For~$x \in T$, we have under assumption~\textbf{A1}
\begin{align*}
    \mathbb{P}(S_{w}(x) < b) \geq 1 - \delta.
\end{align*}
\end{lemma}

\begin{proof}
As~$x \in T$, by definition there is no change point in the interval~$[x-w, x+w]$. Consider the random variables~$\{\widetilde X_i\}_{i = x-w}^{x+w}$. Let~$\mu_x$ denote the mean of the segment. Let~$\Psi(\{\widetilde X_i\}_{i=x-w}^{x-1})$ and~$\Psi(\{\widetilde X_i\}_{i=x+1}^{x+w})$ define the scan statistic in \texttt{RC-Cat}. By Corollary~\ref{cor:bnd.eta}, we have using~$w \geq \frac{B}{A\eta} \log(4/\delta)$, 
\begin{align*}
    |\Psi(\{\widetilde X_i\}_{i=x-w}^{x-1}) - \mu_x| &\leq c_0 \sqrt{\mathcal{M} \eta}, \\
    |\Psi(\{\widetilde X_i\}_{i=x+1}^{x+w}) - \mu_x| &\leq c_0 \sqrt{\mathcal{M} \eta},
\end{align*}
each with probability at least~$1-\delta/2$. Therefore, the event~$\mathcal{E}_1(x)$ occurs with probability at least~$1-\delta$. Indeed, by triangle inequality
\begin{align*}
S_w(x) = |\Psi(\{\widetilde X_i\}_{i=x-w}^{x-1}) - \Psi(\{\widetilde X_i\}_{i=x+1}^{x+w})| &\leq 2c_0 \sqrt{\mathcal{M} \eta}.
\end{align*}
The result holds.
\end{proof}

\begin{lemma} \label{lem:sec}
Let~$\{\widetilde X_i\}_{i \in [n]}$ be a collection that is input to \texttt{RC-Cat}. Let the threshold~$b = 2c_0 \sqrt{\mathcal{M} \eta}$. Let assumptions~\textbf{A1} - \textbf{A3} hold. For~$y \in \mathcal{K}$, we have 
\begin{align*}
    \mathbb{P}(S_{w}(y) > b) \geq 1 - \delta.
\end{align*}
\end{lemma}

\begin{proof}
For any~$k$, consider~$y = \tau_k$. By assumption~\textbf{A1} and \textbf{A2}, we have that the segment~$\{\widetilde X_i\}_{i= \tau_k+1}^{\tau_k+w}$ has mean~$\mu_{\tau_{k}}$ and the segment~$\{\widetilde X_i\}_{i= \tau_k - w}^{\tau_k-1}$ has mean~$\mu_{\tau_{k-1}}$. For simplicity, we abuse the notation and denote~$\Psi(\{\widetilde X_i\}_{i= \tau_k+1}^{\tau_k+w}):= \Psi_1$, and~$\Psi(\{\widetilde X_i\}_{i= \tau_k - w}^{\tau_k-1}): = \Psi_2$.  We have using the inequality $(x+y)^2 \geq x^2/2 - y^2$ for any~$x,y \in \mathbb{R}$,
{{\begin{align*}
    &\Big|\Big(\Psi_1 - \mu_{\tau_{k}}\Big) - \Big(\Psi_2 - \mu_{\tau_{k-1}}\Big)  + (\mu_{\tau_{k}} - \mu_{\tau_{k-1}} )\Big|^2 \\
    & \geq \theta^2/2 - \Big|\Big(\Psi_1 - \mu_{\tau_{k}}\Big) - \Big(\Psi_2 - \mu_{\tau_{k-1}}\Big)\Big|^2.
\end{align*}}}
By Corollary~\ref{cor:bnd.eta}, we have using~$w \geq \frac{B}{A\eta} \log(4/\delta)$, 
\begin{align*}
\Big|\Big(\Psi_1 - \mu_{\tau_{k}}\Big)\Big|^2 &\leq b^2/4,\\
\Big|\Big(\Psi_2 - \mu_{\tau_{k-1}}\Big)\Big|^2 &\leq b^2/4,
\end{align*}
each with probability at least~$1-\delta/2$. The result follows using assumption~\textbf{A3}.
\end{proof}

\newpage

\begin{lemma} \label{lem:fin}
Let the threshold~$b = 2c_0 \sqrt{\mathcal{M} \eta}$. Let assumptions~\textbf{A1} - \textbf{A3} hold. On event~$\mathcal{E}_n$, we have
\[
\{\widehat K = K\}~~\&~~ \max_{k \in \widehat{\mathcal{K}}}|\widehat \tau_k - \tau_k| \leq w.
\]
Moreover, $\mathbb{P}(\mathcal{E}^c_n) \rightarrow 0$.
\end{lemma}

\begin{proof}
First, note that~$\widehat \tau_k \in T^c,~\forall~k \in \widehat{\mathcal{K}}$ by definition of~$T$, where~$\widehat \tau_k$ are the change points detected by \texttt{RC-Cat}. Therefore, we have that
\[
\tau_k \in [\widehat \tau_k - w, \widehat \tau_k + w].
\]
By assumption \textbf{A2}, we have that there are no other change points in this interval. 

Next, we show that there is a change point identified in the interval~$(\tau_k -w, \tau_k + w)$. Let~$\lambda$
The intervals
\[
\Omega_+: = (\tau_k+w,\tau_k+(\lambda+1)w)~~ \&~~ \Omega_-:= (\tau_k-w,\tau_k-(\lambda+1)w)
\]
are contained in~$T$ by definition and \textbf{A2}. This implies that on every~$x \in \Omega_+ \cup \Omega_-$, the event~$\mathcal{E}_1(x)$ holds with the corresponding scan statistic~$S_w(x)<b$. However, by Lemma~\ref{lem:sec} we have $S_w(\tau_k) > b$. This implies that there is a local maximum, say~$\widehat \tau_k \in (\tau_k -w, \tau_k + w)$, and $S_{w}(\widehat \tau_k) \geq S_w(\tau_k) > b$.

Using union bound, we have
\begin{align*}
    \mathbb{P}(\mathcal{E}^c_n) \leq \sum_{k \in \mathcal{K}} \mathbb{P}(\mathcal{E}^c_2(\tau_k)) + \sum_{x \in T} \mathbb{P}(\mathcal{E}^c_1(x)).
\end{align*}
From Lemma~\ref{lem:fir} and Lemma~\ref{lem:sec}, we have the trivial upper bound
\[
\mathbb{P}(\mathcal{E}^c_n) \leq 2n\delta.
\]
The result holds by choosing~$\delta = \frac{1}{n^{c_1+1}}$ for~$c_1>0$.
\end{proof}

\subsection{Proof of Corollary~\ref{cor:no.fa}}
Let~$[m]:= \{t: w+1 \leq t \leq n-w \}$. As~$\mathcal{K} = \emptyset$, we have that the event
\[
\mathcal{E}_n = \cap_{x \in [m]} \mathcal{E}_1(x)
\]
has a probability lower bound using Lemma~\ref{lem:fir} as
\[
\mathbb{P}(\mathcal{E}_n) \geq 1 - n\delta,
\]
where~$\delta = \frac{1}{n^{c_1+1}}$ for~$c_1>0$.

\subsection{Proof of Theorem~\ref{thm:cnt.est}}
The soft-truncation of~$k$ samples as~$\Bar{\mu}_{\eta}$
obtains from Theorem~\ref{thm:dev.bnd}, the deviation bound
\[
|\Bar{\mu}_{\eta} - \mu| \leq \vartheta_k := \sqrt{2 \mathcal{M} \Big( \frac{\log(4/\delta)}{k} + 2 A \eta \Big)},
\]
with probability at least~$1-\delta/2$. For the shifted data~$\widetilde X'_i = \widetilde X_i - \Bar{\mu}_{\eta}$ note that~$\mathbb{E}[\widetilde X'^2|\mathcal{F}] \leq (\mathcal{V}+ \vartheta_k^2 )$. So the soft-truncation estimate of shifted data obtains from Theorem~\ref{thm:dev.bnd}, the  deviation bound
\[
|\mu'_{\eta} - (\mu-\Bar{\mu}_{\eta})| \leq \sqrt{2 (\mathcal{V}+ \vartheta_k^2 )\Big( \frac{\log(4/\delta)}{n-k} + 2 A \eta \Big)},
\]
with probability at least~$1-\delta/2$. Defining~$\mu'_{\eta}:= \widehat \mu_{\eta} -  \Bar{\mu}_{\eta}$, the result follows with probability~$1-\delta$ as both high-probability events should hold.

\begin{figure}[t!]
	\centering
	\mbox{
		\includegraphics[width=2.5in]{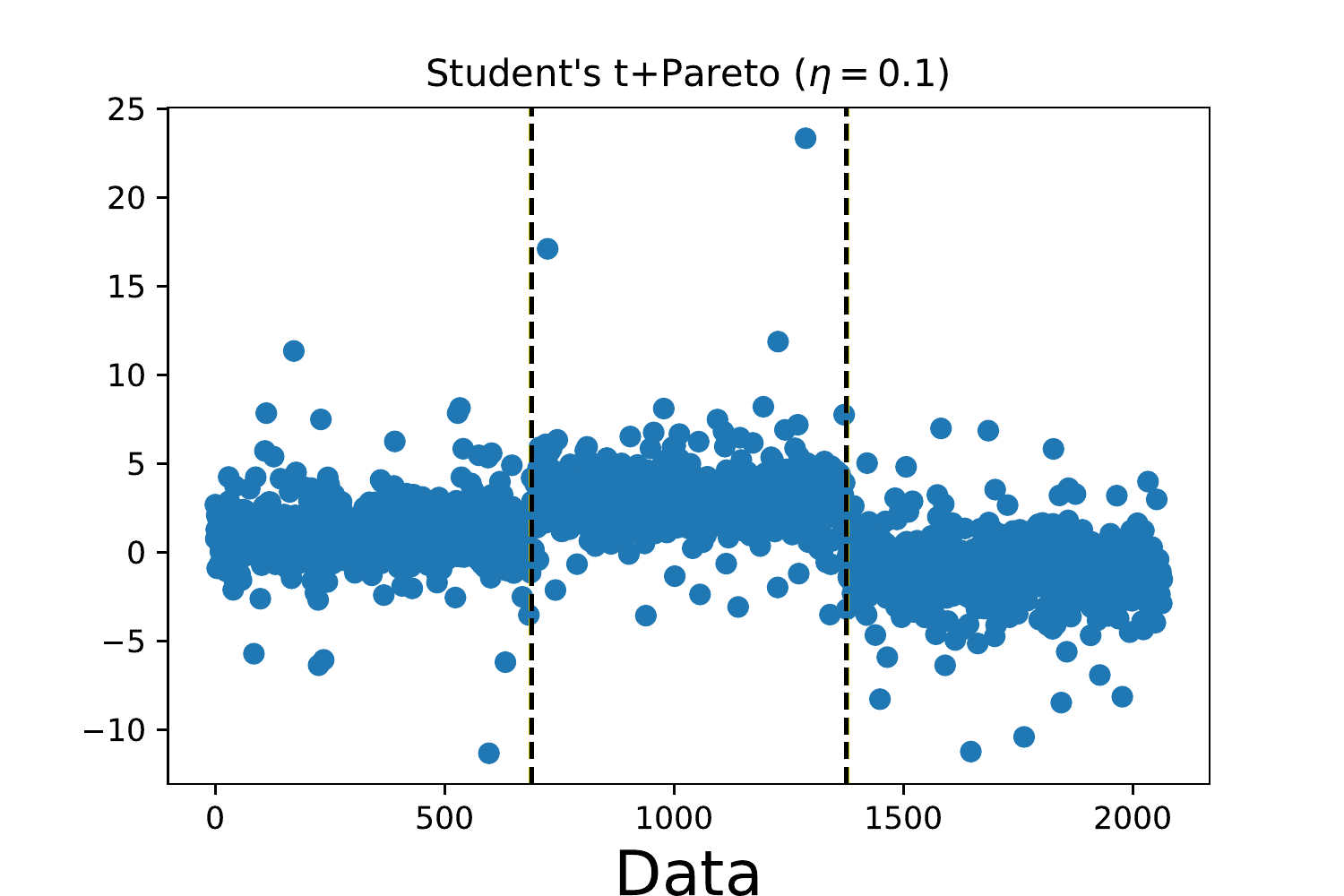}\hspace{0.2in}
		\includegraphics[width=2.5in]{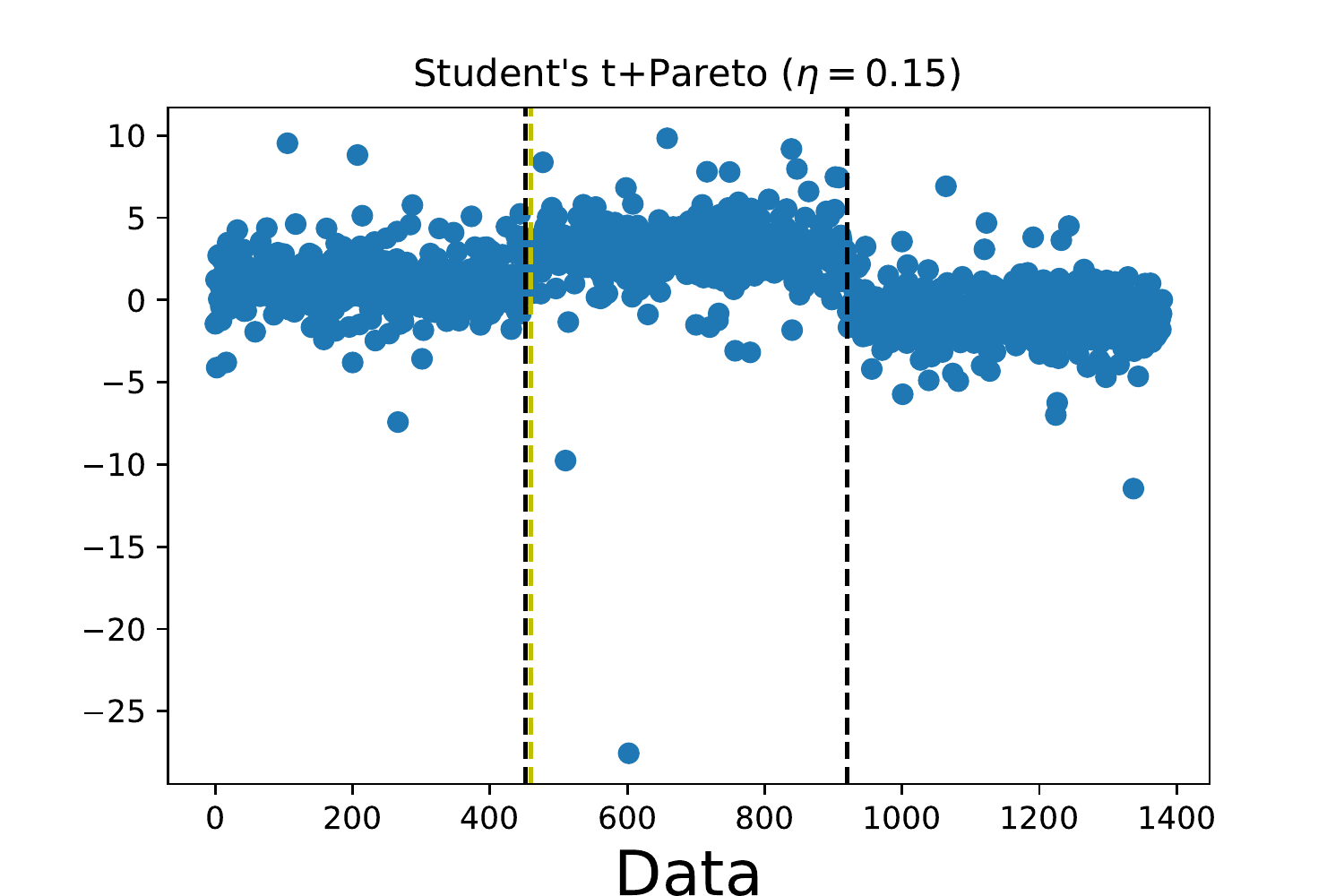} 
	}
	
	\mbox{
		\includegraphics[width=2.5in]{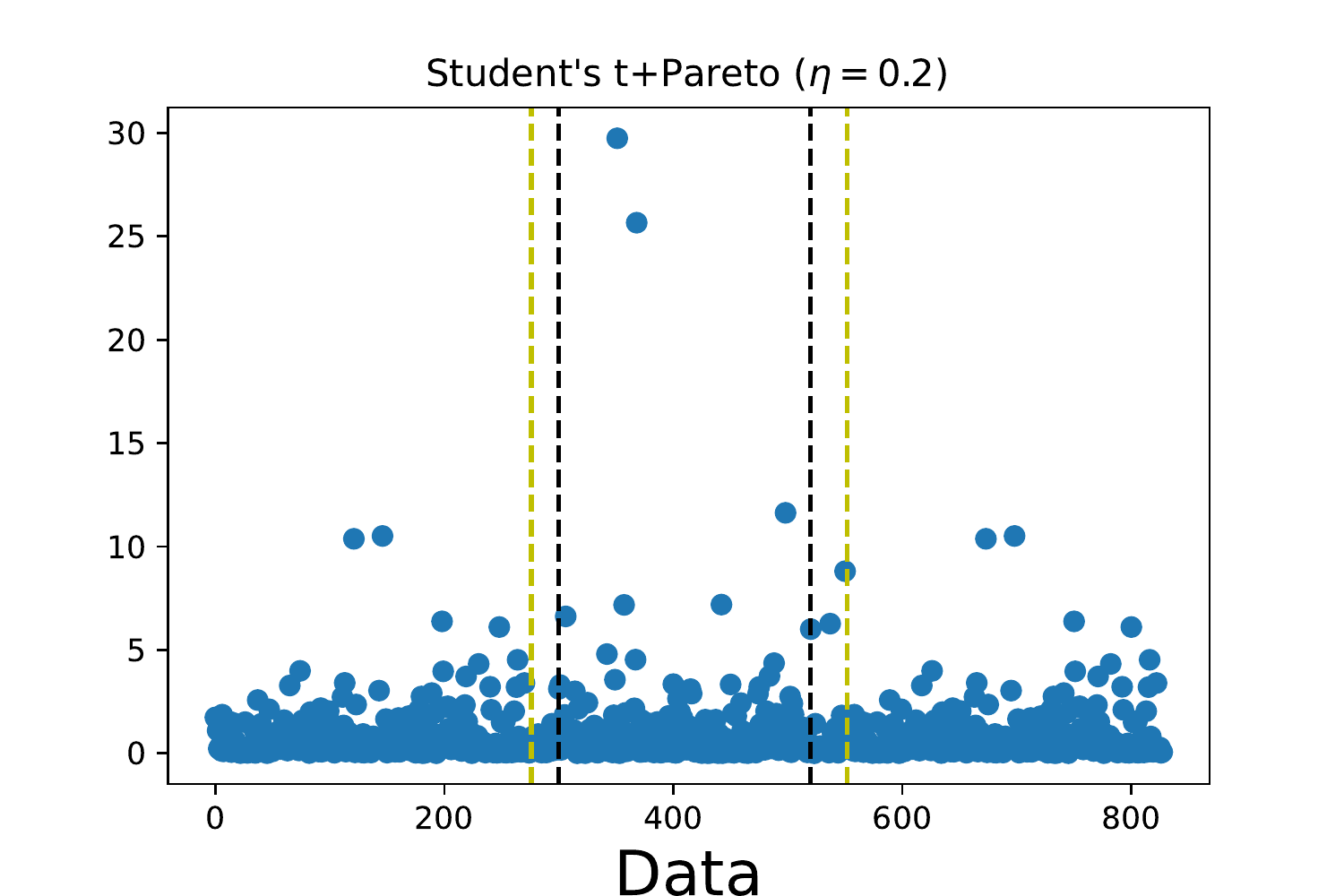}\hspace{0.2in}
		\includegraphics[width=2.5in]{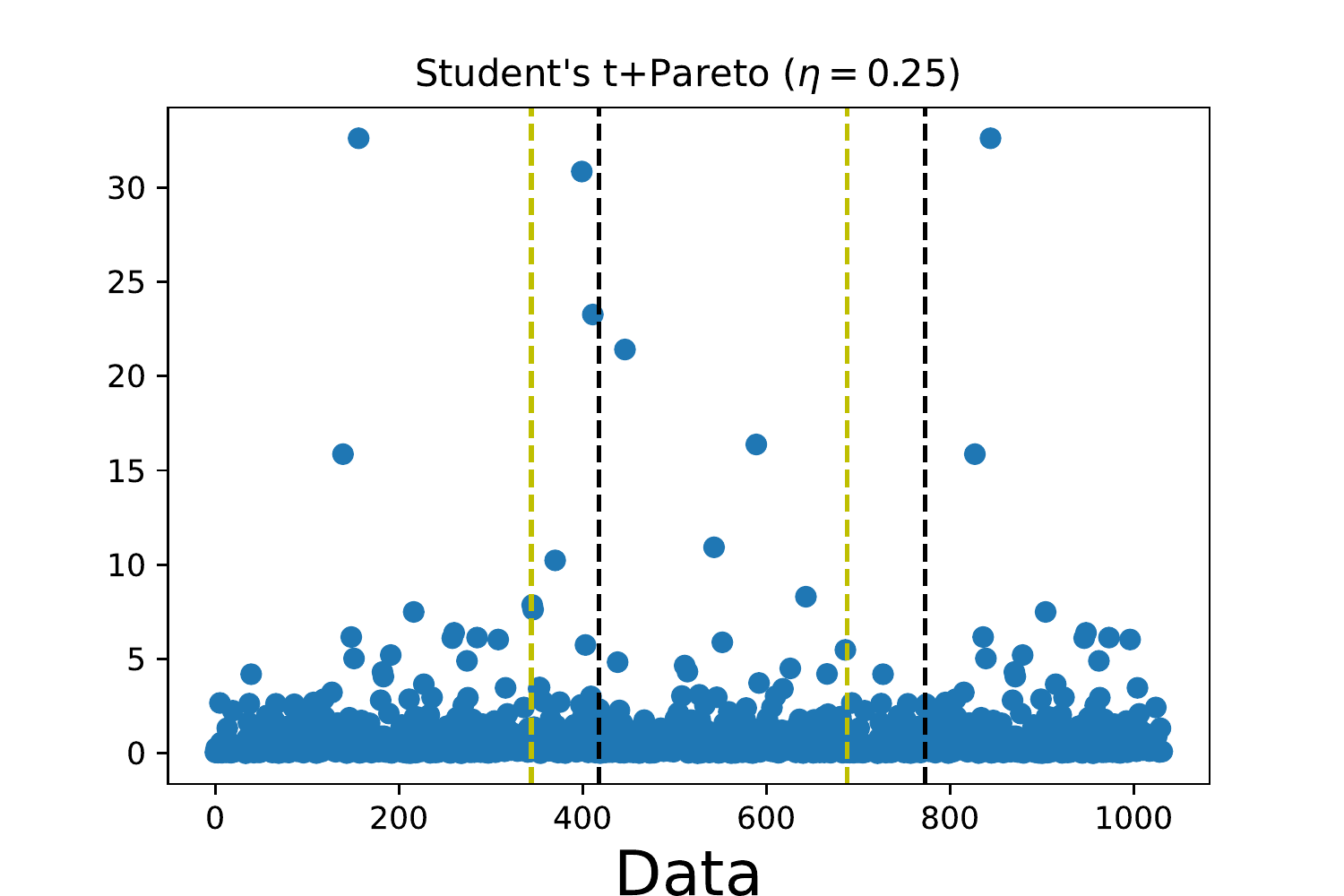} 
	}
	
	\mbox{
		\includegraphics[width=2.5in]{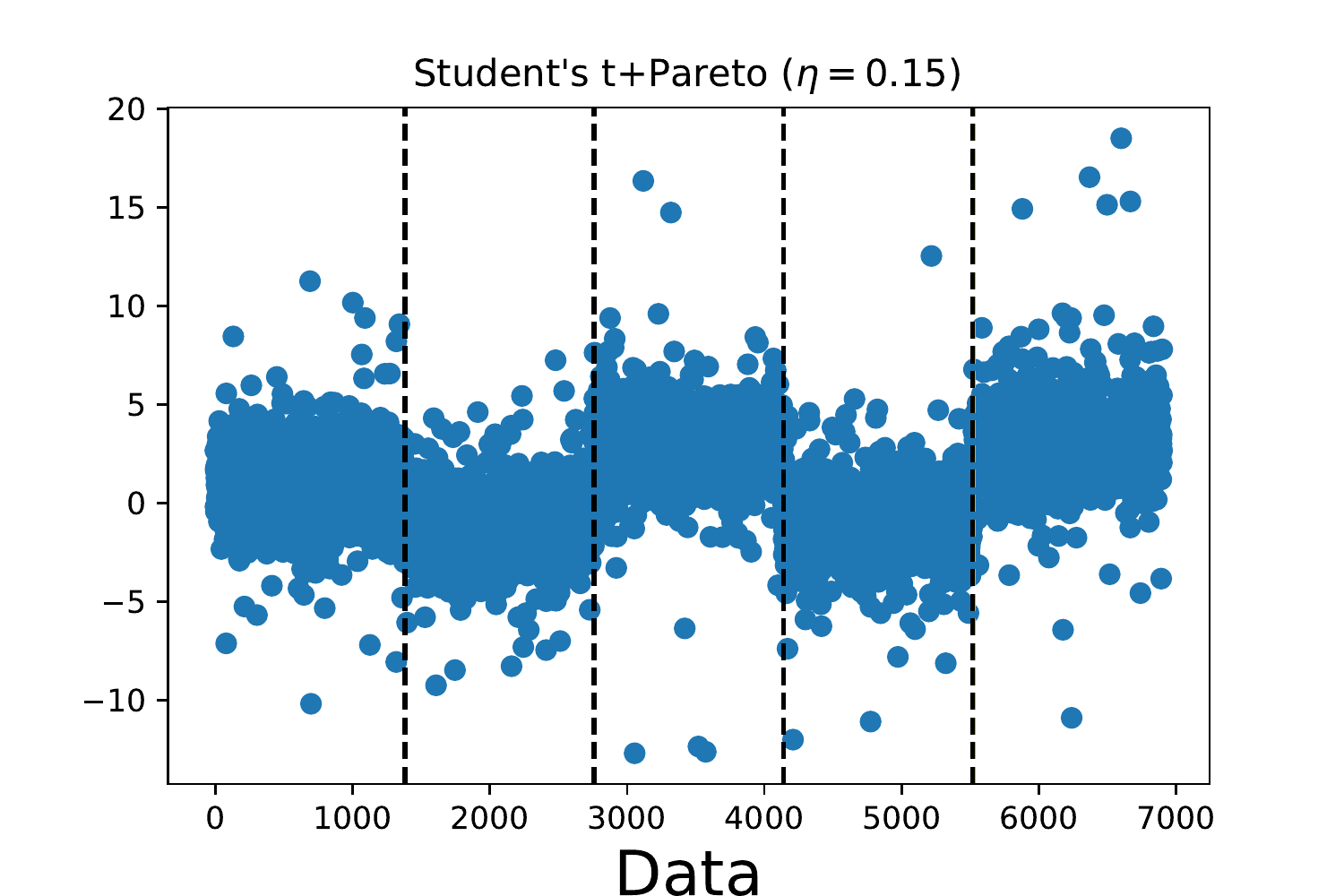}\hspace{0.2in}
		\includegraphics[width=2.5in]{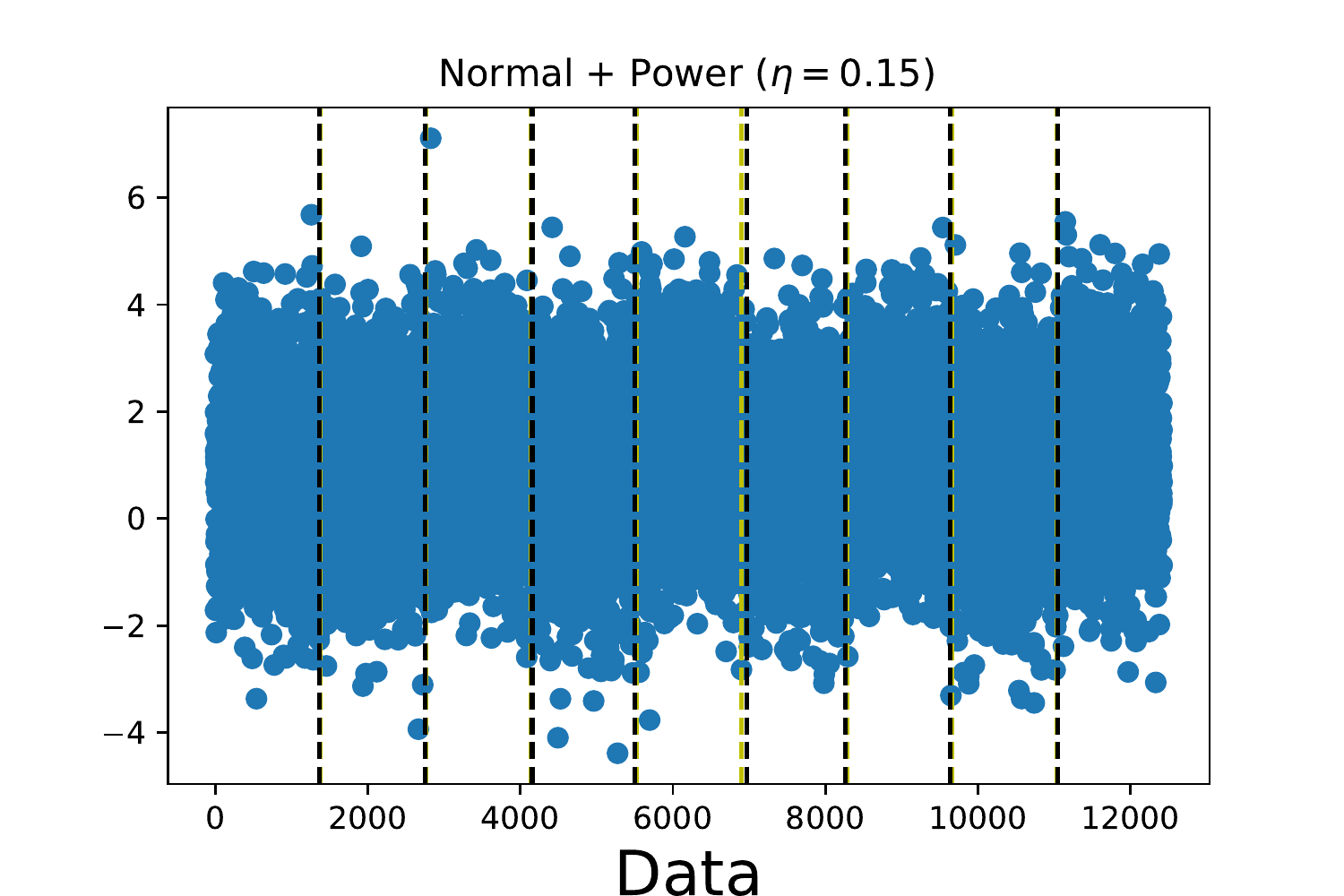} 
	}
\vspace{-0.1in}

	\caption{The figures show the performance of Algorithm~\ref{alg:seq} under different distributions for the outliers. The \textit{yellow} line indicates the positions at which the distribution of the inliers changed, and the~\textit{black} dotted line shows the positions at which changes were announced by the algorithm. It can be seen that while the contamination level~$\eta$ influences the offset -- increases with increase in~$\eta$ -- it has limited bearing on the number of change points detected under a careful choice of the algorithm parameters.}\label{fig:ch_det}	
	
\end{figure}

\section{Numerical Results}
In this section, we provide numerical results to illustrate the performance of Algorithm~\ref{alg:seq}. Our main objective is to provide empirical evidence to support the consistency results. 

\subsection{Synthetic Data}
We assume that the adversary/ nature replaces the original data with samples generated from random distributions. The algorithm parameters for all the figures in Figure~\ref{fig:ch_det} are chosen as follows. The inlier distributions are modeled as
\[
X_t = \mu_t + \zeta_t,
\]
where~$\zeta_t$ is a martingale difference noise. From Theorem~\ref{thm:cnt.est}, the choice of~$w$ is given as~$w \geq \frac{B}{A \eta} \log(4/\delta)$ for a confidence level~$\delta = 0.01$, $B=2$, and $A = \log 2$. The value~$\mu_t \leq 3,~\forall~t$, obtaining a bound~$\mathcal{M} = 10$ for unit variance. There is a trade-off between false detection and no-detection for different choices of~$b$ informed by Theorem~\ref{thm:dev.bnd}. For good performance, we recommend setting smaller than that informed by theory and increasing the neighbourhood width~$\lambda$ for local search and elimination. For Figure~\ref{fig:ch_det},~$\lambda = 3$ was chosen and~$\alpha$ is informed by Theorem~\ref{thm:dev.bnd}.

\subsection{Comparison with ARC method}

In this section, we compare our method with a recent state-of-art method, the ARC algorithm~\citep{LY21}.
Specifically, we examine the robustness of proposed method and ARC under three different contamination settings. 

(Setting 1) \textit{Pareto contamination}.  The inliers follow student-t distribution with degree of freedom 3. Outliers follow pareto distribution with degree freedom 2. 

(Setting 2) \textit{One-sided arbitrary contamination}. The inliers follow student-t distribution with degree of freedom 3. Outliers are fixed at 100. 

(Setting 3) \textit{Two-sided arbitrary contamination}. The inliers follow student-t distribution with degree of freedom 3. Outliers are fixed at 100 or -100.

The total time horizon $T$ is fixed at $1500$, the confidence level $\delta = 0.01$, $A = \log(2), B= 2, \mathcal M = 5$, the true mean is in the range, $-3 \leq \mu_t \leq 3$ and \textit{two} underlying change points equally spaced between~$[0,T]$.
We consider varying the following tuning parameter. The contamination rate $\eta \in \{5\%, 10\%, 20\%, 30\%, 40\%\}$.
Window size $w \in \{80, 100, 120\}$.
Average differences (i.e., average of $|\hat \tau_k - \tau_k|$'s) between detected time and true change points are reported. To be fair (without deliberately tuning threshold $b$), for both methods, the detected change points are chosen to be time stamps with top two $S_w(k)$ values.

\begin{table}[b!]
	\centering
	\begin{tabular}{ c|c|c|c|c|c} 
		\hline
		\hline
		\multicolumn{6}{c}{Setting 1} \\
		\hline 
		\hline
		\multicolumn{6}{c}{$w = 80$} \\
		\hline
		$\eta$ & 0.05 & 0.1& 0.2 & 0.3 & 0.4 \\ 
		\hline 
		Ours  & 6.5 & 13.3 & 31.3 & 43.3 & 55.9  \\
		\hline
		ARC   & 18.5 & 26.1 & 40.9 & 44.2 & 58.6  \\
		\hline 
		\hline
		\multicolumn{6}{c}{$w = 100$} \\
		\hline
		$\eta$ & 0.05 & 0.1& 0.2 & 0.3 & 0.4 \\ 
		\hline 
		Ours  & 3.6 & 6.6 & 14.8 & 26.4 & 31.2  \\
		\hline
		ARC   & 19.8 & 20.2 & 22.6 & 32.7 & 36.9  \\
		\hline 
		\hline
		\multicolumn{6}{c}{$w = 120$} \\
		\hline
		$\eta$ & 0.05 & 0.1& 0.2 & 0.3 & 0.4 \\ 
		\hline 
		Ours  & 2.9 & 3.4 & 7.0 & 7.7 & 17.7  \\
		\hline
		ARC   & 17.7 & 18.0 & 19.9 & 23.5 & 25.6  \\
		\hline 
		\hline
	\end{tabular}
	\caption{The table of detection error under Setting 1 with various choices of tuning parameters $\eta$ and $w$. Each case is replicated for 500 times.} 
	\label{tab1}
\end{table}

\begin{table}[h!]
	\centering
	\small
	\begin{tabular}{ c|c|c|c|c|c} 
		\hline
		\hline
		\multicolumn{6}{c}{Setting 2} \\
		\hline
		\hline 
		\multicolumn{6}{c}{$w = 80$} \\
		\hline 
		$\eta$ & 0.05 & 0.1& 0.2 & 0.3 & 0.4 \\ 
		\hline 
		Ours  & 2.0 & 3.5 & 9.6 & 25.8 & 43.5  \\
		\hline
		ARC   & 15.3 & 14.4 & 19.7 & 134.6 & 46.2  \\
		\hline 
		\hline
		\multicolumn{6}{c}{$w = 100$} \\
		\hline
		$\eta$ & 0.05 & 0.1& 0.2 & 0.3 & 0.4 \\ 
		\hline 
		Ours  & 2.0 & 2.9 & 13.4 & 21.3 & 44.9  \\
		\hline
		ARC   & 13.8 & 14.5 & 34.5 & 119.4 & 98.7  \\
		\hline 
		\hline
		\multicolumn{6}{c}{$w = 120$} \\
		\hline
		$\eta$ & 0.05 & 0.1& 0.2 & 0.3 & 0.4 \\ 
		\hline 
		Ours  & 2.1 & 2.5 & 11.8 & 25.6 & 44.0  \\
		\hline
		ARC   & 16.0 & 14.5 & 20.4 & 81.0 & 84.7  \\
		\hline 
		\hline
	\end{tabular}
		\caption{The table of detection error under Setting 2 with various choices of tuning parameters $\eta$ and $w$.} 
	\label{tab2}\vspace{0.2in}
\end{table}

\begin{table}[h!]
	\centering
	\small
	\begin{tabular}{ c|c|c|c|c|c} 
		\hline
		\hline
		\multicolumn{6}{c}{Setting 3} \\
		\hline
		\hline
		\multicolumn{6}{c}{$w = 80$} \\
		\hline
		& 0.05 & 0.1& 0.2 & 0.3 & 0.4 \\ 
		\hline 
		Ours  & 2.6 & 3.9 & 10.6 & 25.5 & 36.5  \\
		\hline
		ARC   & 14.8 & 14.3 & 13.9 & 81.8 & 130.7  \\
		\hline 
		\hline
		\multicolumn{6}{c}{$w = 100$} \\
		\hline
		& 0.05 & 0.1& 0.2 & 0.3 & 0.4 \\ 
		\hline 
		Ours  & 2.9 & 3.7 & 10.0 & 21.0 & 34.1  \\
		\hline
		ARC   & 15.0 & 13.7 & 18.8 & 97.1 & 96.1 \\
		\hline 
		\hline
		\multicolumn{6}{c}{$w = 120$} \\
		\hline
		& 0.05 & 0.1& 0.2 & 0.3 & 0.4 \\ 
		\hline 
		Ours  & 2.6 & 3.6 & 9.1 & 15.2 & 32.2  \\
		\hline
		ARC   & 15.7 & 14.7 & 14.0 & 70.6 & 62.2  \\
		\hline 
		\hline
	\end{tabular}
		\caption{The table of detection error under Setting 3 with various choices of tuning parameters $\eta$ and $w$.} 
	\label{tab3}
\end{table}

Based on the Tables~\ref{tab1} -~\ref{tab3}, we can find that the proposed method is robust to different contamination level, while ARC method is not. Especially when we increase contamination rate $\eta$ to 40 \%, ARC behaves much worse. 
Moreover, our method is also less sensitive to the choices of window size than ARC method.
These results indicate that~\texttt{RC-Cat} is indeed a better method.

\subsection{Real-world Data}

We consider two real data sets in this subsection, the \textit{well-log} data~\citep{well, FR19, LY21} which has been widely studied in the existing literature and \textit{PM2.5 index} data~\citep{PM25} which has not been considered in the literature.  

\textbf{Well-log} data set contains 4050 measurements of nuclear magnetic response during the drilling of a well. Majority of the observations behave very well and a small proportion of the observations are far away from the mean value.

\textbf{PM2.5 index} data set records air quality of Hong Kong during 1-Jan 2014 to 2-Feb-2022. The PM2.5 index fluctuates occasionally over the total period of time. 

\newpage

From Figure~\ref{fig:realdata}, we can see that the proposed method \texttt{RC-Cat} can well detect the jump points in well-log data and is very robust to those outliers. 
Our method can also capture the fluctuations of Hong Kong PM2.5 index.

\begin{figure}[t]
\centering
	\mbox{
		\includegraphics[width=3in]{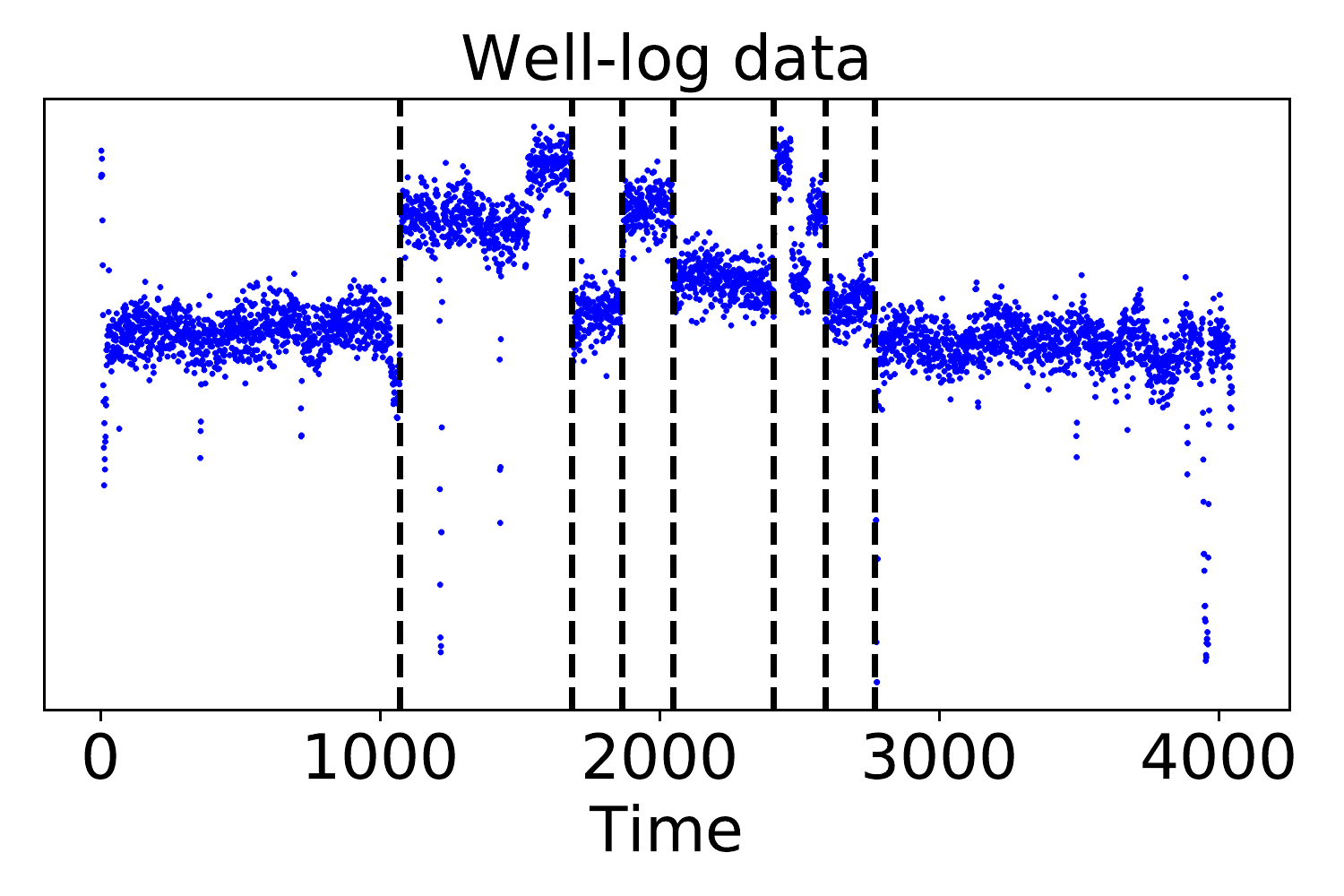}\hspace{0.2in}
		\includegraphics[width=3in]{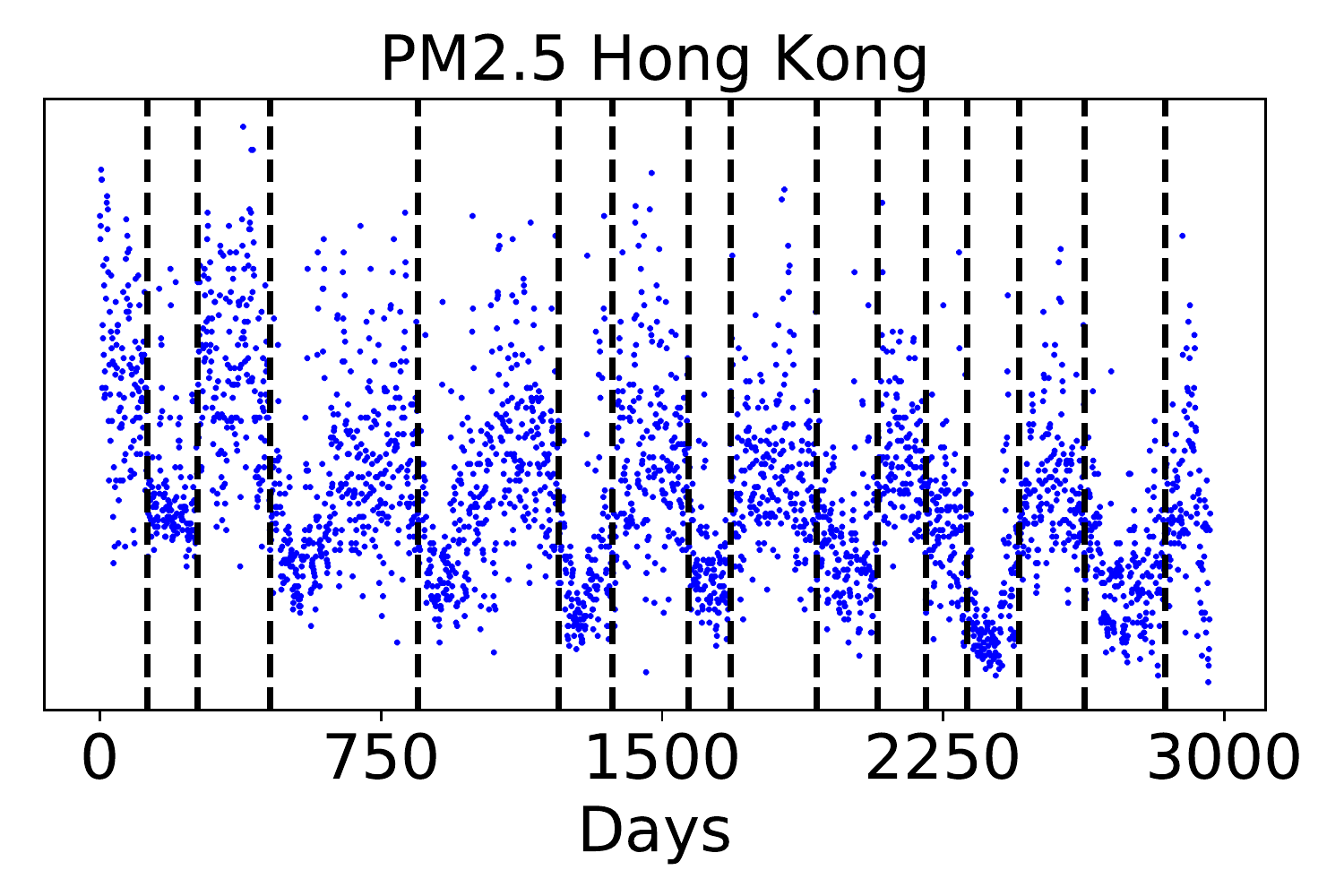} 
	}
\vspace{-0.1in} 
	\caption{The detection results returned by \texttt{RC-Cat} for two real-world data sets, well-log data and PM2.5 index data.
	In the experiment, we choose tuning parameters $b = c_0 \sqrt{\mathcal M \eta}/2$ and $\lambda = 1$.}\label{fig:realdata}	
	
\end{figure}

\section{Conclusion}
In this work, we provided a robust change detection algorithm based on influence functions that can deal with a fraction of arbitrary but weakly adversarial contamination. Key contributions to the vast literature on robust offline change detection methods include: (i)~The ability to handle non-i.i.d data along with contamination, when minimal assumptions are made on the distributions of the inliers. (ii)~A computationally appealing algorithm that is consistent. The algorithm itself is intuitive, and combines local search methods to segment the dataset. Also, empirical results confirm that the algorithm outperforms the state of the art offline change detection algorithm in terms of average detection times, demonstrating significant gains under heavy-contamination.

This work motivates change detection in multi-variate datasets, possibly in the presence of contamination, motivated by the appealing aspect of obtaining dimension-free robust estimation in high-dimension using  influence functions; see~\cite{CG17}. 

\clearpage

\bibliographystyle{plainnat}
\bibliography{refs}

\end{document}